\documentclass[10pt,conference]{IEEEtran}
\IEEEoverridecommandlockouts

\usepackage{cite}
\usepackage{amsmath,amssymb,amsfonts, amsthm}
\usepackage{mathdots}
\usepackage[linesnumbered, noend, ruled]{algorithm2e}
\usepackage{graphicx}
\usepackage{textcomp}
\usepackage{xcolor,colortbl}
\usepackage{enumerate}
\usepackage[pdfborder={0 0 0}]{hyperref}
\newcommand{\Fq}{ \mathbb{F}_{q}}
\newcommand{\y}{\boldsymbol{y}}
\newcommand{\f}{\boldsymbol{f}}
\newcommand{\bb}{\boldsymbol{b}}
\newcommand{\e}{\boldsymbol{e}}
\newcommand{\bphi}{\boldsymbol{\varphi}}
\newcommand{\rank}{\text{rank}}
\newcommand{\df}{d_f}
\newcommand{\dg}{d_g}
\newcommand{\dfge}{d_{fgE}}
\newcommand{\deltafgE}{\delta_{fgE}}
\newcommand{\dA}{d_A}
\newcommand{\db}{d_b}
\newcommand{\alphab}{\boldsymbol{\alpha}}

\newtheorem{theorem}{Theorem}
\newtheorem{definition}[theorem]{Definition}

\newtheorem{remark}[theorem]{Remark}

\def\BibTeX{{\rm B\kern-.05em{\sc i\kern-.025em b}\kern-.08em
    T\kern-.1667em\lower.7ex\hbox{E}\kern-.125emX}}
    
\begin{document}

\setlength{\abovedisplayskip}{3pt}
\setlength{\belowdisplayskip}{3pt}

\title{Enhancing simultaneous rational function recovery: adaptive error correction capability 
and new bounds for applications
}

\author{\IEEEauthorblockN{Eleonora Guerrini, Romain Lebreton, Ilaria Zappatore}
\IEEEauthorblockA{
\textit{LIRMM, Universit\'e de Montpellier, CNRS}\\
Montpellier, France \\
\textbf{\{guerrini, lebreton, zappatore\}@lirmm.fr}}}

\maketitle

\begin{abstract}
In this work we present some results that allow to improve the decoding radius in solving polynomial linear systems with errors in the scenario where errors are additive and randomly distributed over a finite field.  The decoding radius depends on some bounds on the solution that we want to recover, so their overestimation could significantly decrease our error correction capability. For this reason, we introduce an algorithm that can bridge this gap, introducing some {\it ad hoc} parameters that reduce the discrepancy between the estimate decoding radius and the effective error correction capability.

\end{abstract}
\section{Introduction}
The family of Reed Solomon codes (RS)  is a large class of very well studied algebraic codes. 
They are MDS codes, they perform list-decoding   and have efficient decoding algorithms that can be viewed in a computer algebra setting as  rational reconstruction problems.  
More specifically, we are interested on Interleaving Reed Solomon (IRS) codes. 
IRS codes are well studied and can be decoded efficiently by a bounded distance (BD) decoder beyond the unique correction capability radius for almost all error patterns (\emph{cf.}~\cite{bleichenbacher_decoding_2003},\cite{schmidt_collaborative_2009},\cite{puchinger_decoding_2017}). 

In this work we focus on the problem of solving a polynomial linear system with errors (PLSwE), introduced in \cite{boyer_numerical_2014} and \cite{kaltofen_early_2017}. 
Since the solution of PLSwE   is a vector of rational function, the PLSwE is a special case of the problem of reconstructing a vector of rational function given its evaluations, some of which could be erroneous (the simultaneous rational function recovery, shortly SRFR).
In \cite{guerrini_polynomial_2019} we proposed an algorithm, based on IRS decoding, that allows to solve SRFR (and in particular PLSwE), correcting more than \cite{boyer_numerical_2014} and \cite{kaltofen_early_2017} in a probabilistic setting.
In this paper, we improve the technique of \cite{guerrini_polynomial_2019}, increasing the error correction capability.
Since we want to recover a vector of rational functions $\y$, which is solution of a polynomial linear system over a finite field, $A(x) \y(x) = \bb(x)$, we introduce a new bound on the error correction capability which also depends on the bounds on the degrees of $A$ and $\bb$.
Moreover,  the knowledge of the degrees of the solution would allows us to reach an \textit{ideal} error correction capability, but we do not know these degrees and their overestimation could significantly decrease the amount of errors we could correct. 
For this reason, we introduce a \textit{parameter oblivious} algorithm for the PLSwE that allows us to get closer to the ideal error correction capability, without knowing the real degrees. 

The paper is structured as follows: in Section~\ref{IRS} we recall standard facts for IRS codes, in Section~\ref{polyLinSystErrSection} we introduce our problem (PLSwE) and we set up the model. In Section~\ref{SR} we introduce the generalization of the PLSwE, \textit{i.e} the simultaneous rational function recovery (SRFR). We present a technique, based on IRS decoding, that allows to achieve a bigger error correction capability. In Section~\ref{ET}, we propose our algorithm and in Section~\ref{technicalResultsSec}, we present our main theorem, the cornerstone of all our technical results. Finally in Section~\ref{ConclusionSec}, we expose our open problems and conclusions.

\section{Interleaved Reed-Solomon codes}
\label{IRS}

	A Reed-Solomon (RS) code of length $n$ and dimension $k$ over $\Fq$ can be defined as the set 
	$
	\mathcal{C}_{RS}(n,k)=\{(f(\alpha_1), \ldots, f(\alpha_n)) \mid f\in \Fq[x], \deg(f)\leq k-1\}
	$
	where $\alphab:=\{\alpha_1,\ldots,\alpha_n\}$ is the set of distinct evaluation points over $\Fq$.
 For $l \geq 1 $, 
	an $l$-Interleaved Reed-Solomon (IRS) code is defined by
	the direct sum of $l$ RS codes $\mathcal{C}_{RS}(n, k_i)$ sharing the same set of evaluation points, \textit{i.e.}
	$$
	\mathcal{C}_{IRS}(n,\boldsymbol{k})=\left \{
(\boldsymbol{c}_i)_{1 \leq i \leq l}
	\in (\Fq)^{l\times n}\mid \boldsymbol{c}_i \in \mathcal{C}_{RS}(n,k_i)\right\} 
	$$
	If $k=k_1=\ldots=k_l$, we say that an IRS is homogeneous and we denote it $\mathcal{C}_{IRS}(n,k;l)$.
From now, we will focus only on homogeneous IRS codes. 
Codewords in $\mathcal{C}_{IRS}(n,k;l)$ can be seen as evaluations of  $\f(x)=(f_1,\ldots,f_l)\in (\Fq[x])^{l\times1}$ with $\deg(\f):=\max_{1\leq i \leq l}(\deg(f_i))\leq k-1$.

We now consider the decoding instance
$
\Upsilon=C+\Xi \in (\Fq)^{l\times n}
$
where $C \in \mathcal{C}_{IRS}(n,k;l)$ and $\Xi$ is the error matrix. We can see $\Upsilon=(\y_j)_{1\leq j\leq n}$ and $\Xi=(\e_j)_{1\leq j \leq n}$ both in $(\Fq^l)^n$. As error model we consider \textit{burst errors}, \textit{i.e.} the error positions are the nonzero columns of the error matrix $\Xi$. In detail, for any $1\leq j\leq n$, the set of error positions is $E:=\{1\leq j\leq n \mid \e_j\neq \boldsymbol{0}\}$. The number of errors is then $|E|$.

Since $C=(\f(\alpha_1), \ldots, \f(\alpha_n))$ for $\f\in (\Fq[x])^{l\times 1}$, $\deg(\f)\leq k-1$, for any $1\leq j\leq n$,
\begin{equation}{\label{IRSeq}}
\y_j=\f(\alpha_j)+\e_j.
\end{equation}
In order to decode $C$, we need to recover the vector of polynomials $\f(x)$.

For  $1\leq i \leq l$, let $Y_i\in \Fq[x]$ be the \textit{Lagrangian polynomials} 
such that $Y_i(\alpha_j)=y_{ij}$, for $1\leq j\leq n$, and $\deg(Y_i)<n$ and let  $\Lambda=\prod_{j\in E}(x-\alpha_j)$ be the \textit{error locator polynomial}. We observe that for any $1\leq i\leq l$,
$\Lambda Y_i \equiv \Lambda f_i \mod \prod_{j=1}^n(x-\alpha_j)$.
This is a nonlinear equation in the unknowns $\Lambda(x)$ and $\f(x)$. A classic approach for decoding RS codes (\emph{cf.}~\cite{gao_new_2003}, \cite{elwyn_r._berlekamp_error_1986}), that  can be extended to IRS codes, consists in the \textit{linearization} of these equations, by replacing $\Lambda(x)$ and $\Lambda(x) f_i(x) $ with the unknowns $\lambda(x)$ and $\varphi_i(x)$. In this way we obtain the \textit{key equation}
\begin{equation}{\label{IRSKeyEq}}
\lambda Y_i \equiv \varphi_i\mod \prod_{j=1}^n(x-\alpha_j).
\end{equation}
In order to decode, it suffices to study the set $S$ of $(\lambda, \varphi_1,\ldots, \varphi_l)\in \Fq^{(l+1)\times 1}$ which verify (\ref{IRSKeyEq}) and such that $\deg(\lambda)\leq |E|$ and $\deg(\varphi_i)\leq |E|+k-1$.
IRS codes, can be decoded by efficient BD decoders beyond the unique decoding radius. These decoders succeed \textit{for almost all} error patterns \cite{bleichenbacher_decoding_2003}. With ``almost all" we mean that there exists a polynomial $R$ such that the decoder succeeds for all instances $\Upsilon$ satisfying $R(\Upsilon) \neq 0$.
A quite tight estimation of the probability of failure can be founded in \cite{schmidt_collaborative_2009} and improvements on the decoding radius recently appeared in  \cite{puchinger_decoding_2017}.

In the next section we will remark the parallel between the problem of solving polynomial linear systems with errors and the IRS decoding. (\emph{cf.}~\cite{guerrini_polynomial_2019}).

\section{Polynomial Linear System Solving With Errors}\label{polyLinSystErrSection}
Given $l\geq 1$, we study the problem of solving a consistent full rank polynomial linear system over a finite field $\Fq$,
\begin{equation}\label{PolyLinSyst}
A(x)\y(x)=\bb(x)
\end{equation}
where  
$A(x)\in (\Fq[x])^{l\times l}$ is full rank and $\bb(x)\in (\Fq[x])^{l\times 1}$.
Any solution of this system is a vector of rational functions, \textit{i.e.} $\y(x)=\left(\frac{\hat{f}_1(x)}{g_1(x)}, \ldots, \frac{\hat{f}_l(x)}{g_l(x)}\right)\in \Fq(x)^{l\times 1}$. Let $g(x)$ be the monic least common denominator, then there is a unique solution 
\begin{equation}\label{uniqueSol}
\y(x)=\left(\frac{f_1(x)}{g(x)}, \ldots, \frac{f_l(x)}{g(x)}\right) \in \Fq(x)^{l\times 1}
\end{equation}
that is also \textit{reduced}, \textit{i.e.} $\gcd(f_1(x),\dots,f_l(x),g(x))=1$. Our main aim is to reconstruct such a solution. Note that this common denominator representation can be more compact and it appears frequently for solutions of linear systems computed by the Cramer's rule.

As in \cite{boyer_numerical_2014}, \cite{kaltofen_early_2017}, \cite{guerrini_polynomial_2019}, we will analyze a scenario where some \textit{errors} occur.
In detail, we fix $n$ evaluation points $\alphab$ and we suppose that any evaluation point is not a root\footnote{In \cite{boyer_numerical_2014} and \cite{kaltofen_early_2017}, the authors study a more general case. 
They fix $n$ distinct evaluation points, without any assumptions about the roots of $g(x)$. In our work, we need this assumption in order to prove our results.} of the polynomial $g(x)$. 
In our model, there is a black box which, for any evaluation point $\alpha_j$, gives a solution $\y_j \in (\Fq)^{l\times 1}$ of the evaluated system of linear equations\footnote{We suppose that for any evaluation points $\alpha_j$, the rank of the evaluated matrix $A(\alpha_j)$ still remains full. 
In \cite{boyer_numerical_2014} and \cite{kaltofen_early_2017} there was also studied the rank drop case.}$A(\alpha_j) \y_j=b(\alpha_j)$.
However this black box could do some errors in the computations.

\begin{definition}{(Erroneous evaluation points \cite{boyer_numerical_2014})}\label{errEvPointsDef}

An evaluation point $\alpha_j$ is erroneous iff
$
\y_j\neq \frac{\f(\alpha_j)}{g(\alpha_j)}.
$
We denote by $E:=\left\{1\leq j \leq n \mid \y_j\neq \frac{\f(\alpha_j)}{g(\alpha_j)}\right\}$ the set of positions of the erroneous evaluations.
\end{definition}
We can now formalize our problem,
\begin{definition}{(Polynomial linear system solving with errors)}\label{PLSwE}

The problem of solving a polynomial linear system  with errors (denoted PLSwE) consists in recovering the vector of rational functions (\ref{uniqueSol}), \textit{i.e.} the unique solution of a consistent, full rank polynomial linear system (\ref{PolyLinSyst}), given
\begin{itemize}
    \item $n$ evaluation points $\alphab$,
    \item $\df\geq \deg(\f)$, $\dg\geq \deg(g)$, $\dA\geq \deg(A)$, $\db\geq \deg(\bb)$,
    \item the black box output $(y_{ij})_{\substack{1\leq i \leq l\\1\leq j\leq n}}$,
    \item a bound on the number of erroneous evaluation points $\varepsilon\geq |E|$.
\end{itemize}
\end{definition}{}
\begin{remark}\label{PLSwEEquivFracDec}
We observe that if $j\in E$, then
there exists a nonzero $\e_j\in (\Fq)^{l\times 1}$ such that $\y_j= {\f(\alpha_j)}/{g(\alpha_j)}+\e_j$.

In general, for any $1\leq j\leq n$,
$\y_j={\f(\alpha_j)}/{g(\alpha_j)}+\e_j$ where $\e_j\neq \boldsymbol{0}$ iff $j\in E$.
\end{remark}

We can conclude that the PLSwE can be seen as the \textit{extension} of the problem of decoding an IRS code (see (\ref{IRSeq})) to the rational function case.
 
\section{simultaneous rational function recovery}\label{SR}
\begin{definition}{(Simultaneous rational function recovery)}\label{df:SRFR}

Fix some parameters $n, q, \df, \dg, \varepsilon, \alphab$ such that
$0\leq \df, \dg, \varepsilon < n\leq q$.
An instance of the simultaneous rational function recovery problem (shortly SRFR) is a matrix $(\y_j)_{1\leq j\leq n}=(y_{ij})_{\substack{1\leq i \leq l\\1\leq j \leq n}} \in (\Fq)^{l\times n}$ such that there exists
\begin{itemize}
    \item a reduced vector of rational functions $\frac{\f(x)}{g(x)} \in (\Fq(x))^{l\times 1}$, where $\deg(\f)\leq \df$, $\deg(g)\leq \dg$ and $\forall j, \ g(\alpha_j) \neq 0$.
    \item a matrix 
    $
    (\e_j)_{1\leq j \leq n}=
    (e_{ij})_{\substack{1\leq i\leq l\\1\leq j\leq n}}$
    such that its \emph{column support} $E:=\{1\leq j\leq n \mid \e_{j}\neq \boldsymbol{0}\}$ satisfies  $|E|\leq \varepsilon$;
\end{itemize}
which satisfy
$
\y_{j}={\f(\alpha_j)}/{g(\alpha_j)}+\e_{j}.
$
The solution of an SRFR instance is $(\f(x),g(x))$.
\end{definition}{}
This problem was introduced in \cite{pernet_high_performance_2014} and
in \cite{guerrini_polynomial_2019}. We now present a recovering algorithm in the model of IRS codes. 
In detail, let $(\y_{j})_{1\leq j \leq n}$ be an instance of the SRFR problem with parameters $ n, \df, \dg, \varepsilon$ (we will omit $q$ and $\alphab$ for simplicity).
As for IRS codes we introduce
the error locator polynomial $\Lambda=\prod_{j\in E}(x-\alpha_j)$ and the Lagrangian polynomials $(Y_i(x))_{1\leq i \leq l}$.
We observe that $\f,g, \Lambda$ still satisfy 
$
\Lambda g Y_i \equiv \Lambda f_i  \mod \prod_{j=1}^n(x-\alpha_j)
$
so, as for IRS codes, we study the \textit{key equations}
\begin{equation}{\label{OurKeyEq}}
\begin{array}{c}
 \psi Y_i \equiv \varphi_i \mod \prod_{j=1}^n(x-\alpha_j) \text{ for } 1\leq i \leq l,\\
\end{array}
\end{equation}
whose unknowns are the polynomials $\varphi_i(x)$ and $\psi(x)$ such that $\deg(\varphi_i)\leq d_f+\varepsilon$ and $\deg(\psi)\leq d_g+\varepsilon$.
We observe that this is equivalent to study the evaluated system 
\begin{equation}{\label{evKeyEq}}
[\varphi_i(\alpha_j)=y_{ij}\psi(\alpha_j) ]_{\substack{ 1\leq i \leq l\\1\leq j\leq n}} .
\end{equation}
In this case the unknowns are the $d_f+\varepsilon+1$ coefficients of any $\varphi_i(x)$ and the $d_g+\varepsilon+1$ coefficients of $\psi(x)$.

Let $S_{\y, \df+\varepsilon, \dg+\varepsilon}$ be the $\Fq$-vector space of  $(\bphi,\psi)=(\varphi_1,\ldots,\varphi_l,\psi)\in (\Fq[x])^{(l+1)\times 1}$ which verify (\ref{OurKeyEq}) and the degree constraints $deg(\bphi)\leq \df+\varepsilon$ and $deg(\psi)\leq \dg+\varepsilon$.

\begin{remark}{\label{kernelRemark}}
Since we can consider the key equations in the polynomial (\ref{OurKeyEq}) or evaluated version (\ref{evKeyEq}), studying the solution space $S_{\y, \df+\varepsilon, \dg+\varepsilon}$ is equivalent to study the right kernel of the coefficient matrix $M_{\y, \df+\varepsilon, \dg+\varepsilon}$ (see \cite{bleichenbacher_decoding_2003, brown_probabilistic_2004}) of the evaluated key Equation~\eqref{evKeyEq}.
In detail,
\begin{equation}{\label{coeffMatrix}}
\small
M_{\y,\df+\varepsilon, \dg+\varepsilon}=\begin{pmatrix}V_{\df+\varepsilon+1} &       & &-D_1V_{\dg+\varepsilon+1}\\
                                         &\ddots & &\vdots\\
                                         &  &V_{\df+\varepsilon+1} &-D_lV_{\dg+\varepsilon+1}\\
\end{pmatrix}{}
\end{equation}
where $V_t=(\alpha_j^{i-1})_{\substack{1\leq j\leq n\\1\leq i \leq t}}$ is an $n\times t$ \textit{Vandermonde} matrix
 and $D_i$ is the matrix with $y_{i1},\ldots,y_{in}$ on the diagonal.

\end{remark}{}

\begin{theorem}(\emph{cf.}~\cite{boyer_numerical_2014})\label{thKalto}
	If 
	$
	\varepsilon\leq \frac{(n-d_f-d_g-1)}{2}=:\varepsilon_{BK},
	$
	then all solutions $(\bphi,\psi)\in S_{\y, \df+\varepsilon, \dg+\varepsilon}$ lead to the same vector of rational functions $\f/g$, \textit{i.e.}
	$
	\frac{\bphi(x)}{\psi(x)}=\frac{\f(x)}{g(x)}.
	$
\end{theorem}
This means that below this error correction capability $\varepsilon_{BK}$, all the elements $(\bphi,\psi)\in S_{\y, \df+\varepsilon, \dg+\varepsilon}$ are polynomial multiples of the unique solution $(\Lambda\f,\Lambda g)$.
Besides, it is possible to prove that 
$
 S_{\y, \df+\varepsilon, \dg+\varepsilon}=\langle x^i\Lambda \f ,x^i\Lambda g \rangle_{0\leq i \leq d_{fgE}}   
$
where
\begin{equation}{\label{deltafgeDef}}
 d_{fgE}= \min(\df-\deg(\f), \dg-\deg(g))+\varepsilon-|E| .
\end{equation}
Hence we can uniquely reconstruct the vector of rational functions $\f/g$.
We observe that if $\df=\deg(\f)$ (or $\dg=\deg(g)$) and $|E|=\varepsilon=\varepsilon_{BK}$
 the solution space $S_{\y,\df+\varepsilon,\dg+\varepsilon}$ is a vector space of dimension $1$, spanned by  $(\Lambda \f ,\Lambda g)$.

In \cite{guerrini_polynomial_2019}, motivated by the analogy of the SRFR problem and the decoding of IRS codes, we proved that, under some assumptions on the error distribution, we can correct more than $\varepsilon_{BK}$ errors with a certain probability. 

We now set up our probabilistic model.
We focus on $(\y_{j})_{ 1\leq j \leq n }$, an instance of the SRFR  with parameters $n, \df, \dg, \varepsilon$ with \textit{random errors}.
In detail, we suppose that
$
\y_{j}={\f(\alpha_j)}/{g(\alpha_j)}+\e_{j},
$
where ${\f(x)}/{g(x)}$ is a reduced fraction with degrees bounded by $\df, \dg$ and such that $g(\alpha_j) \neq 0$.
Moreover $\e_j$ is uniformly distributed in
$(\Fq)^{l\times 1}$ if $j \in E$ and $\e_{j} = \boldsymbol{0}$ if $j\notin E$, for a fixed error position set $E$ with $|E| \leq \varepsilon$.
Under this assumption we have the following.

\begin{theorem}{(\emph{cf.}~\cite{guerrini_polynomial_2019})}
Fix $n,\df,\dg$ and $\varepsilon_{GLZ}=\frac{l(n-\df-\dg-1)}{l+1}$.
Let 
$(\y_j)_{1\leq j\leq n}$
be an instance of the SRFR with random errors and parameters $ n, \df, \dg=\deg(g), \varepsilon=|E|=\lfloor\varepsilon_{GLZ}\rfloor$. Then the corresponding solution space $S_{\y,\df+|E|,\deg(g)+|E|}$ is a vector space of dimension $1$, spanned by the solution $(\Lambda \f ,\Lambda g )$, with probability at least $1-\frac{\deg(g)+|E|}{q}$.
\end{theorem}{}

In this work we extend the previous result of \cite{guerrini_polynomial_2019} to the general case where we only know a bound $\dg\geq \deg(g)$
and a bound $\varepsilon$ on the number of errors $|E|$, with $\varepsilon \leq \varepsilon_{GLZ}=\frac{l(n-\df-\dg-1)}{l+1}$.

\begin{theorem}{\label{ourThm1}}
Fix $n,\df,\dg$ and $\varepsilon \leq \varepsilon_{GLZ}$ and take $\dfge$ as in (\ref{deltafgeDef}). Let 
$(\y_j)_{1\leq j \leq n}$
be an instance of SRFR with random errors and parameters $n, \df, \dg, \varepsilon$. Then with probability $\geq 1-\frac{\dg+\varepsilon}{q}$
 we get
$S_{\y, \df+\varepsilon, \dg+\varepsilon}=\langle x^i\Lambda \f ,x^i\Lambda g \rangle_{0\leq i \leq d_{fgE}}   
$.
\end{theorem}

\noindent\emph{A. Simultaneous rational function recovery of a solution of a polynomial linear system with errors.}
By Remark~\ref{PLSwEEquivFracDec} we can deduce that the PLSwE coincides with SRFR with parameters $ n, \df, \dg, \varepsilon$. The matrix $(\y_{j})_{1\leq j\leq n}$, which is the black box output, is then an instance of this problem. Hence, all the results of the previous section hold.
Furthermore, since we want to reconstruct a vector of rational functions that is a solution of a polynomial linear system, it is possible to introduce a bound on the error correction capability which depends on the bounds on the degree of the polynomial matrix $A(x)$ and on $\bb(x)$ (as shown in \cite{kaltofen_early_2017}).

Let $(\y_{j})_{1\leq j\leq n}$ be an instance of PLSwE with parameters $n, \df, \dg, \dA, \db$.
Recall from the previous section that $S_{\y,\df+\varepsilon, \dg+\varepsilon}$ is the set of $(\bphi,\psi)$ which verify (\ref{OurKeyEq}) and such that $\deg(\bphi)\leq \df+\varepsilon$ and $\deg(\psi)\leq \dg+\varepsilon$.

\begin{theorem}(see~\cite{kaltofen_early_2017})\label{thCabay}
  Let $\varepsilon_{KPS} := \frac{n-\max(\dA+\df, \db+\dg)-1}{2}$.
  If $\varepsilon \leq \varepsilon_{KPS}$, then $S_{\y,\df+\varepsilon, \dg+\varepsilon}=\langle x^i\Lambda \f ,x^i\Lambda g \rangle_{0\leq i \leq d_{fgE}}$.
  
  The same result holds if we consider $\varepsilon\leq \max(\varepsilon_{BK},\varepsilon_{KPS})$.
\end{theorem}
There are some cases in which this error correction capability is bigger than $\varepsilon_{BK}$.
In fact, as proved in \cite{kaltofen_early_2017}, when all the bounds are tight, $\varepsilon_{BK}<\varepsilon_{KPS}$ iff $\deg(g(x))>\deg(A(x))$.

In this paper, we will introduce a new bound on the error correction capability based on $\dA$ and $\db$ under probabilistic assumptions. 
In particular, given a polynomial linear system as in (\ref{PolyLinSyst}), we suppose that the black box returns $(\y_j)_{1\leq j\leq n}$ 
where $\y_j$ is uniformly distributed in $(\Fq)^{l\times 1}$ if $j\in E$ (instead of $\y_j\neq \f(\alpha_j)/g(\alpha_j)$).
By Remark~\ref{PLSwEEquivFracDec},
$
\y_j=\frac{\f(\alpha_j)}{g(\alpha_j)}+\e_j
$
and so our probabilistic assumption on the black box output is indeed an assumption on the error distribution, \textit{i.e.} $\e_j$ is uniformly distributed in $(\Fq)^{l\times 1}$ (instead of $\e_j \neq \boldsymbol{0}$), when $j\in E$.
We will call PLSwE with random errors, the PLSwE in this error model.
\begin{theorem}{\label{ourThm2}}
  Fix $n, \df, \dg ,\dA, \db$, take $\dfge$ as in (\ref{deltafgeDef}) and,
  $$\varepsilon \leq \varepsilon_{GLZ2}:= \frac{l(n-\max(\dA+\df,\db+\dg)-1)}{l+1}.$$
  Let 
  $(\y_j)_{1\leq j\leq n}$
  be an instance of PLSwE with random errors with parameters $ n, \df, \dg, \dA, \db$. Then with probability at
  least $ 1-\frac{dg+\varepsilon}{q}$ we get $S_{\y,\df+\varepsilon, \dg+\varepsilon}=\langle x^i\Lambda \f ,x^i\Lambda g \rangle_{0\leq i \leq d_{fgE}}$.
  
Thus the same result holds when $\varepsilon \leq \max(\varepsilon_{GLZ}, \varepsilon_{GLZ2})$.
\end{theorem}{}

We will prove Theorems~\ref{ourThm1} and~\ref{ourThm2} in Section˜ \ref{technicalResultsSec}.

\section{Parameter oblivious decoding algorithm}\label{ET}
The PLSwE problem takes as input some degree bounds $\df$ and $\dg$, hence all our error correction capabilities until now depend on these bounds, \textit{e.g.} $\varepsilon_{GLZ}=\frac{l(n-\df-\dg-1)}{l+1}$. Most importantly, our technique for solving the PLSwE requires such degree bounds to decode up to this capability.

Ideally we could decode up to 
$\frac{l(n-\deg(\f)-\deg(g)-1)}{l+1}$ errors by taking the bounds tight, \textit{i.e.} $\df = \deg(\f), \dg = \deg(g)$. Our lack of knowledge of the real degrees $\deg(\f), \deg(g)$ limits us to correct this ideal amount of errors. Indeed, the bounds $\df$ and $\dg$ could overestimate the degrees of $\f(x)$ and $g(x)$, thus significantly decreasing all our error correction capability bounds. 

In this work, we propose a \textit{parameter oblivious} (\cite{khonji_output-sensitive_2010}, \cite{pernet_high_performance_2014}) algorithm that allows to get closer to the ideal error correction capability even without the knowledge of the real degrees.

In \cite{khonji_output-sensitive_2010} the authors already observed that even for classic RS codes (Section~\ref{IRS}), the knowledge of a bound instead of the real degree of $f$, could decrease the error correction capability. They proposed an algorithm for standard RS codes that allows to correct up to $\frac{n-\deg(f)-1}{2}\geq \frac{n-k}{2}$ errors.

On the other hand, in \cite{kaltofen_early_2017} it was introduced an algorithm for solving the PLSwE up to 
$
|E|\leq \max(\varepsilon'_{BK}, \varepsilon'_{KPS})
$
where
$$
\begin{array}{l}
\varepsilon'_{BK}:=\frac{n-\max(\deg(\f)+\dg, \deg(g)+\df)-1}{2}\geq \varepsilon_{BK},\\
 \varepsilon'_{KPS}:=\frac{n-\max(\dA+\deg(\f)),\db+ \deg(g))-1}{2}\geq \varepsilon_{KPS}
\end{array}
$$

In this work, we propose an algorithm that succeeds for almost all instances
$(\y_j)_{1\leq j \leq n}$
of a PLSwE with parameters $n, \df, \dg, \dA, \db, \varepsilon$
whenever
$
|E|\leq \max(\varepsilon'_{GLZ}, \varepsilon'_{GLZ2})
$
where 
$$
\small
\begin{array}{l}
    \varepsilon'_{GLZ}:=
    n-\max(\deg(\f)+\dg, \df+\deg(g))-\lceil \frac{\varepsilon}{l} \rceil-1 \geq \varepsilon_{GLZ}\\
    \varepsilon'_{GLZ2}:=
    n-\max(\dA+\deg(\f),\db+\deg(g))-\lceil \frac{\varepsilon}{l} \rceil-1 \geq \varepsilon_{GLZ2}
\end{array}{}
$$
We will explain later where these bounds come from.
Our new capability $\varepsilon'_{GLZ}$ can be greater than $\varepsilon'_{BK}$, especially when $\varepsilon$ is a tight bound on the number of errors. 
In particular, if we assume that $|E| \leq \varepsilon'_{GLZ}$, or equivalently if $\varepsilon = \varepsilon'_{GLZ}$, then $\varepsilon'_{GLZ}$ becomes
$$\varepsilon'_{GLZ} = \frac{l(n-\max(\deg(\f)+\dg, \deg(g)+\df)-1)}{l+1} \geq \varepsilon'_{BK}.$$
The same holds for $\varepsilon'_{GLZ2}$ w.r.t. $\varepsilon'_{KPS}$.

The main idea consists in the introduction of some others parameters $\delta_f, \delta_g, \xi$ and on the study of the solution space of the key equation~\ref{OurKeyEq} with new degree constraints $\delta_f+\xi$ and $\delta_g+\xi$.
As in (\ref{deltafgeDef}) we define,
\begin{equation}{\label{eq:deltafgE}}
  \deltafgE:=\min(\delta_f-\deg(\f), \delta_g-\deg(g))+\xi-|E|.  
\end{equation}

Informally speaking, we will see in the following theorem that the introduction of these new degree constraints will allow us to increase the error correction capability.

\begin{algorithm}[tb]\label{parameterObliviousAlgo}
 \KwData{
  $(y_{ij})_{\substack{1\leq i\leq l\\1\leq j \leq n}}$, an instance of PLSwE with parameters $ n, \df, \dg, \dA, \db, \varepsilon$}

\KwResult{$(\bphi,\psi)$ (equal to $(\Lambda \f, \Lambda g)$ with high probability) or
 ``$|E|>\max(\varepsilon'_{GLZ}, \varepsilon'_{GLZ2})$"}

\bigskip
\label{l:firstaffect}$\delta_f+\xi \leftarrow n-\dg-\lceil \frac{\varepsilon}{l} \rceil-1$;\ $\delta_g+\xi \leftarrow n-\df-\lceil \frac{\varepsilon}{l} \rceil-1$\\

Let $S_{\y,\delta_f+\xi, \delta_g+\xi}$ be the solution space of the key equation~\ref{OurKeyEq} with degree constraints $\delta_f+\xi$, $\delta_g+\xi$. 

  \If{$S_{\y,\delta_f+\xi, \delta_g+\xi}\neq \{(\boldsymbol{0},0)\}$ \label{l:firstif}}
  {\Return $(\bphi, \psi)$ the non zero element of $S_{\y,\delta_f+\xi, \delta_g+\xi}$ with minimal degrees}
\label{l:secaffect}$\delta_f+\xi \leftarrow n-\dA-\lceil \frac{\varepsilon}{l} \rceil-1$; \ $\delta_g+\xi \leftarrow n-\db-\lceil \frac{\varepsilon}{l} \rceil-1$\\
    \If{$S_{\y,\delta_f+\xi, \delta_g+\xi}\neq \{(\boldsymbol{0},0)\}$ \label{l:secif}}
    {\Return $(\bphi, \psi)$ the non zero element of $S_{\y,\delta_f+\xi, \delta_g+\xi}$ with minimal degrees}
    \Return ``$|E|>\max(\varepsilon'_{GLZ}, \varepsilon'_{GLZ2})$";
  
  \caption{Parameter Oblivious Algorithm}\label{parameterObliv}
\end{algorithm}

\begin{theorem}(Parameter oblivious algorithm)
  \label{th:poa}
 
 If $|E| \leq \max(\varepsilon'_{GLZ}, \varepsilon'_{GLZ2}) $ then Algorithm~{\bf \ref{parameterObliv} } outputs $(\bphi,\psi)$ for all instances $(\y_j)_{1\leq j\leq n}$ of the PLSwE. Moreover  $(\bphi,\psi)=(\Lambda \f, \Lambda g)$ with probability $\geq 1-\frac{2(\dg+\varepsilon)}{q}$.

If $|E| > \max(\varepsilon'_{GLZ}, \varepsilon'_{GLZ2}) $, then  Algorithm~{\bf \ref{parameterObliv} } returns  ``$|E|>\max(\varepsilon'_{GLZ}, \varepsilon'_{GLZ2})$" with probability $\geq 1-\frac{2(\dg+\varepsilon)}{q}$.
\end{theorem}

The fact that the algorithm can (probabilistically) detect if $|E|$ exceeds the error correction capability could be used inside another algorithm that would dynamically increase the redundancy $n$ by requesting evaluation on new points, (\textit{cf.} \cite{kaltofen_early_2017}, Algorithm 4.1). 

\begin{remark}
In order to compute the nonzero minimal degree solution (\textit{e.g.} in line 7 of Algorithm~\ref{parameterObliv}), we can use two different approaches: \cite{kaltofen_early_2017} uses column echelon form of the basis of the $\ker(M_{\y,\delta_f+\xi, \delta_g+\xi})$ whereas \cite{rosenkilde-algorithms-sim_2016} proposes $\Fq[x]$-module techniques.
The latter approach yields the best complexity, \textit{i.e.}   $O^{\sim}(l^{\omega-1} n)$ arithmetic operations in $\Fq$ where $\omega < 2.38$ is the linear algebra exponent.
\end{remark}{}

\section{Technical results}\label{technicalResultsSec}

\begin{theorem}\label{th:early}
	Fix $\delta_f, \delta_g, \xi \geq 0$ and let $\deltafgE$ as in (\ref{eq:deltafgE}).
	
Let 
$(\y_j)_{1\leq j\leq n}$
be an instance of the PLSwE with parameters $ n, \deg(\f), \deg(g), |E|, \deg(A), \deg(\bb)$, where $n\geq \min(N_1,N_2)$ and 

\begin{itemize}
	\item $N_1:=\max(\delta_f+\deg(g), \delta_g+\deg(\f))+\xi+\left\lceil{|E|}/{l}\right\rceil+1$,
	\item $N_2:=\max(\delta_f+\deg(A), \delta_g+\deg(\bb))+\xi+\left\lceil{|E|}/{l}\right\rceil+1$,
\end{itemize}{}

Then, with probability at least $1-\frac{\delta_g+\xi}{q}$ we have that
\begin{equation}{\label{solSpace}}
S_{\y,\delta_f+\xi,\delta_g+\xi} =
\langle x^i\Lambda\f,x^i\Lambda g \rangle_{0\leq i \leq \delta_{fgE}}
\end{equation}

By convention, if $\delta_{fgE}<0$, we set $\langle x^i\Lambda\f,x^i\Lambda g \rangle_{0\leq i \leq \delta_{fgE}} =\{(\boldsymbol{0},0)\}$.
\end{theorem}

\begin{proof}
We start by proving that there exists a PLSwE instance $(\y_j)_{1\leq j\leq}$ with the same parameters such that Eq.~\eqref{solSpace}.
 We take a partition $E = \sqcup_{i=1}^l I_i$ with the constraint $|I_i| \leq \lceil |E|/l \rceil$ (it exists since $l \lceil |E|/l \rceil \geq |E|$). For $j \in E$, we define $1 \leq i(j) \leq l$ as the unique index such that 
 $j \in I_{i(j)}$.
 
We separate two cases to prove that $\bphi(x)g(x)=\f(x)\psi(x)$. First if $\min(N_1,N_2)=N_1$,
     then for all $j \in E$ we choose $\y_j\in (\Fq)^{l \times 1}$ such that 
$
    \f(\alpha_j) - g(\alpha_j) \y_j = \boldsymbol{\nu}_{i(j)},
$
 where $\boldsymbol{\nu}_{i(j)}\in (\Fq)^{l \times 1}$ is a vector whose $i(j)$-entry is $1$ and all the others are zero.
 We multiply by $\psi(\alpha_j)$ and we get
$
    \psi(\alpha_j) \f(\alpha_j) - g(\alpha_j) \psi(\alpha_j) \y_j = \psi(\alpha_j) \boldsymbol{\nu}_{i(j)}.
$ By key Equation~\eqref{evKeyEq} we can replace $\psi(\alpha_j) \y_j$ by $\bphi(\alpha_j)$.
 Fix $i$, then
 $
 \forall j \notin I_i, \ 
 \psi(\alpha_j) f_i(\alpha_j) - g(\alpha_j) \varphi_i(\alpha_j)  = 0.
 $
The number of roots of the polynomial  $\psi(x) f_i(x) - g(x) \varphi_i(x)$ is then
 $n-|I_i| \geq n - \lceil |E|/l \rceil \geq \max(\delta_g + \deg (\f), \delta_f + \deg(g)) +
 \xi +1$. Hence, since this polynomial has more roots than its degree it is the zero polynomial.
 
 Second, if $min(N_1,N_2)=N_2$,
 then for all $j \in E$ we choose $\y_j$ such that
 $\f(\alpha_j) - g(\alpha_j) \y_j =-A(\alpha_j)^{-1}g(\alpha_j) \boldsymbol{\nu}_{i(j)}$
 or equivalently
 $A(\alpha_j)\y_j-A(\alpha_j){\f(\alpha_j)}/{g(\alpha_j)}=\boldsymbol{\nu}_{i(j)}$.
 Since $A(\alpha_j)\frac{\f(\alpha_j)}{g(\alpha_j)}=b(\alpha_j)$, after multiplying by $\psi(\alpha_j)$ and using the key Equation~\eqref{evKeyEq} we get
 $A(\alpha_j)\bphi(\alpha_j)-\psi(\alpha_j)b(\alpha_j)=\psi(\alpha_j)\boldsymbol{\nu}_{i(j)}$.
 Fix $i$, then $\forall j \notin I_i, \ (A(\alpha_j)\bphi(\alpha_j)-\psi(\alpha_j)\bb(\alpha_j))_i=0$, \textit{i.e.} the $i$-th component of the polynomial vector $A(x)\bphi(x)-\psi(x)\bb(x)$ vanishes on those $\alpha_j$.
 As before the number $n-|I_i|\geq n-\lceil |E|/l\rceil\geq \max(\delta_f+\deg(A), \delta_g+\deg(A))+\xi+1$ roots of the polynomial $(A(x)\bphi(x)-\psi(x)\bb(x))_i$, is greater than its degree and so it is the zero polynomial.
 We have that $A(x)\f(x)-g(x)\bb(x)=0$ and $A(x)\bphi(x)-\psi(x)\bb(x)=0$. So if we multiply the first equation by $\psi(x)$, the second by $g(x)$ and we subtract we get $A(x)[\bphi(x)g(x)-\f(x)\psi(x)]=0$. Now, $A(x)$ is full rank and so $\bphi(x)g(x)-\f(x)\psi(x)=0$.

Hence in both cases, since $\f/g$ is a reduced fraction, there exists $R \in \Fq[x]$ s.t.
 $\bphi = R \f$ and $\psi = R g$. Going back to Eq.~\eqref{evKeyEq},
 for all $j \in E$ and $i=i(j)$, we get
 $ 0 = \varphi_i(\alpha_j) - \psi(\alpha_j) y_{ij}
      = R(\alpha_j) \left( f_i(\alpha_j) - g(\alpha_j) y_{ij} \right) = R(\alpha_j)$.
Therefore $\Lambda(x)$ divides $R(x)$ and
$(\bphi,\psi) \in \langle x^i \Lambda \f, x^i \Lambda g \rangle$. The
power $i$ must verify
$i + |E| + \deg(\f) = \deg(x^i \Lambda \f) \leq \delta_f + \xi$ and
the same for $g$ which implies exactly that  $i \leq \delta_{fgE}$.
      
Let's now prove that Eq.~\eqref{solSpace} holds with high probability.
We always have
$\langle x^i\Lambda\f, x^i\Lambda g \rangle  \subseteq \ker(M_{\y, \delta_f+\xi, \delta_g+\xi}) = S_{\y,\delta_f+\xi,\delta_g+\xi}$ and Eq.~\eqref{solSpace} is the equality case.
By the rank--nullity theorem, we always have $
\rank(M_{\y,\delta_f+\xi, \delta_g+\xi})\leq \rho
$
where $\rho := n(\delta_f+\xi+1)+\delta_g+\xi-\delta_{fgE}$ and Eq.~\eqref{solSpace} is equivalent to $\rank(M_{\y,\delta_f+\xi, \delta_g+\xi}) = \rho$.
In the first part of the proof, we have proved that there exists an instance $(\y_j)_{1\leq j \leq n}$ such that $\rank(M_{\y, \delta_f+\xi, \delta_g+\xi})=\rho$, which means that there exists a nonzero $\rho$-minor.
If we consider this $\rho$-minor as a polynomial in the variables $(y_{ij})$, we have shown that it is non zero. 
Note that it has total degree at most $\delta_g+\xi$ because only the last $\delta_g+\xi$ columns of $M_{\y,\delta_f+\xi, \delta_g+\xi}$ contain variables $(y_{ij})$ (see Eq.~\eqref{coeffMatrix}).
Therefore the Schwartz-Zippel lemma implies that it cannot be zero in more than $\frac{\delta_g+\xi}{q}$ fraction of its domain. For those instances $(\y_j)_{1\leq j \leq n}$ that don't cancel this $\rho$-minor, we get that $\rank(M_{\y,\delta_f+\xi, \delta_g+\xi})$ is equal to $\rho$ and Eq.~\eqref{solSpace} holds.
\end{proof}

\begin{remark}
Let $(\y_{j})_{1\leq j\leq n}$ be an instance of the PLSwE with parameters $n,\df,\dg,\varepsilon,\dA,\db$. If we consider $\delta_f=\df$, $\delta_g=\dg$ and $\xi=\varepsilon$, then $N_1= \df+\dg+\lceil\frac{(l+1)}{l}\varepsilon \rceil+1$ and $N_2=\max(\dA+\df,\db+\dg)+\lceil\frac{(l+1)}{l}\varepsilon\rceil+1$.
Hence, if $n\geq \min(N_1,N_2)$, by Theorem~\ref{th:early}, the solution space $S_{\y,\df+\varepsilon,\dg+\varepsilon}=\langle x^i \Lambda \f, x^i \Lambda g\rangle$ for $0\leq i \leq \deltafgE$ with probability at least $1- \frac{\dg+\varepsilon}{q}$.

Equivalently, we can fix the number of evaluation points and let the error correction capability vary. 
Thus, let
$(\y_j)_{1\leq j\leq n}$ be an instance
of a PLSwE with parameters $n,\df, \dg,\dA, \db, \varepsilon \leq \max(\varepsilon_{GLZ}, \varepsilon_{GLZ2})$
where
\begin{itemize}
    \item $\varepsilon_{GLZ}=\frac{l(n-\df-\dg-1)}{l+1}$,
    \item $\varepsilon_{GLZ2}=\frac{l(n-\max(\dA+\df, \db+\dg)-1)}{l+1}$.
\end{itemize}
Then, by Theorem~\ref{th:early}, the solution space $S_{\y,\df+\varepsilon,\dg+\varepsilon}=\langle x^i \Lambda \f, x^i \Lambda g\rangle$ for $0\leq i \leq \deltafgE$ with probability at least $1- \frac{\dg+\varepsilon}{q}$. 
Hence we have proved the Theorem~\ref{ourThm1} and Theorem~\ref{ourThm2}.

\end{remark}{}

\begin{proof}[Proof of Theorem~\ref{th:poa}]

First we prove that $|E|\leq \max(\varepsilon'_{GLZ},\varepsilon'_{GLZ2})$ iff there exists a choice of parameters in lines~\ref{l:firstaffect} and ~\ref{l:secaffect} such that $\langle x^i\Lambda \f,x^i\Lambda g\rangle_{0\leq i \leq \deltafgE}\neq \{(\boldsymbol{0}, 0)\}$. We observe that $\langle x^i\Lambda \f,x^i\Lambda g\rangle_{0\leq i \leq \deltafgE}\neq \{(\boldsymbol{0}, 0)\}$ is equivalent to $\deltafgE \geq 0$.
We suppose that $\varepsilon\leq \varepsilon'_{GLZ}$ (we can do the same in the other case) and consider the first choice of $\delta_f+\xi$ and $\delta_g+\xi$ as in line~\ref{l:firstaffect}. Hence, $\varepsilon\leq \varepsilon'_{GLZ}$ iff  $\max(\deg(\f)+\dg, \deg(g)+\df)\leq n-|E|-\lceil \frac{\varepsilon}{l} \rceil-1$. So, $\varepsilon\leq \varepsilon'_{GLZ}$
$$
\Leftrightarrow \begin{cases}
\deg(\Lambda \f) = \deg(\f)+|E| \leq n-\dg-\lceil \frac{\varepsilon}{l} \rceil-1 = \delta_f+\xi\\
\deg(\Lambda g) = \deg(g) +|E| \leq n-\df-\lceil \frac{\varepsilon}{l} \rceil-1 = \delta_g+\xi
\end{cases}
$$
Hence, this is equivalent to $\deltafgE\geq 0$.

Now if $|E|\leq \max(\varepsilon'_{GLZ},\varepsilon'_{GLZ2})$, the latter claim implies 
$\{(\boldsymbol{0}, 0)\} \neq \langle x^i\Lambda \f,x^i\Lambda g\rangle_{0\leq i \leq \deltafgE}\subseteq S_{\y,\delta_f+\xi, \delta_g+\xi}$
and Algorithm~{\bf \ref{parameterObliv}} always outputs $(\bphi, \psi)$.

Second, we claim that for both choices of parameters $\delta_f+\xi,\delta_g+\xi$ (lines~\ref{l:firstaffect} and \ref{l:secaffect}) we have $n\geq \min(N_1,N_2)$.
In fact, if $\delta_f+\xi = n-\dg-\lceil \frac{\varepsilon}{l} \rceil-1,\delta_g+\xi= n-\df-\lceil \frac{\varepsilon}{l} \rceil-1$ then $n\geq N_1 \geq min(N_1,N_2)$. The same holds for the other affectation.
The probability that both solution spaces $S_{\y,\delta_f+\xi, \delta_g+\xi}$ of lines~\ref{l:firstif},~\ref{l:secif} are equal to $\langle x^i\Lambda \f,x^i\Lambda g\rangle_{0\leq i \leq \deltafgE}$ is at least $1-\frac{2(\dg+\varepsilon)}{q}$
by applying Theorem~\ref{th:early} on two different affectations.

Therefore, we can conclude that if $|E|\leq \max(\varepsilon'_{GLZ},\varepsilon'_{GLZ2})$, then with probability at least $1-\frac{2(\delta_g+\xi)}{q}$, $S_{\y,\delta_f+\xi, \delta_g+\xi}$ is equal to $\langle x^i\Lambda \f,x^i\Lambda g\rangle_{0\leq i \leq \deltafgE}$ and since $\deltafgE\geq 0$, the minimal non zero element is $(\varphi, \psi)=(\Lambda \f, \Lambda g)$.

On the other hand, if $|E|>\max(\varepsilon'_{GLZ},\varepsilon'_{GLZ2})$, then $\deltafgE<0$ for both affectations, and so $\langle x^i\Lambda \f,x^i\Lambda g\rangle_{0\leq i \leq \deltafgE}=\{(\boldsymbol{0}, 0)\}$. But with probability at least $1-\frac{2(\delta_g+\xi)}{q}$, both solution spaces $S_{\y,\delta_f+\xi, \delta_g+\xi}$ are equal to $\langle x^i\Lambda \f,x^i\Lambda g\rangle_{0\leq i \leq \deltafgE}=\{(\boldsymbol{0}, 0)\}$ so the algorithm will output ``$|E|>\max(\varepsilon'_{GLZ}, \varepsilon'_{GLZ2})$"
\end{proof}{}

\section{Conclusion and Future Work}\label{ConclusionSec}
In this work, we improve the result of \cite{guerrini_polynomial_2019} considering new bounds on the parameters and taking into account the degrees of $A$ and $\bb$ (as in (\ref{PolyLinSyst})). We also present a parameter oblivious algorithm that allows us to correct more errors.
Our algorithm is probabilistic and the failure probability depends on the parameters $\delta_g$ and $\xi$. Remark that our bound on the failure probability is similar to the original result of\cite{bleichenbacher_decoding_2003} for IRS codes. Actually, this bound on the decoding failure of IRS codes was strongly improved in \cite{schmidt_collaborative_2009}.
Since the SRFR coincides with the reconstruction of a vector of rational functions by its evaluations, some of which erroneous, we can see the SRFR as the decoding of an \textit{interleaved rational code} \cite{pernet_high_performance_2014}.
Despite the similarity of this problem with the decoding of IRS codes, here we deal with a code which is not linear. This prevent the adaptation of most recent techniques for bounding the probability failure of IRS decoding algorithms.
A future work is to provide a better comprehension of the interleaved rational code in order to better bound the failure probability.

\bibliographystyle{IEEEtran}
\bibliography{IEEEabrv,bibliography.bib}

\end{document}